\documentclass[11pt,a4paper]{article}
\usepackage[utf8]{inputenc}
\usepackage[T1]{fontenc}
\usepackage{mathptmx}
\usepackage[margin=1in]{geometry}
\usepackage{graphicx}
\usepackage{booktabs}
\usepackage{amsmath,amssymb,amsthm}
\usepackage{hyperref}
\usepackage{xcolor}
\usepackage{enumitem}
\usepackage{caption}
\usepackage{tabularx}
\usepackage{float}

\hypersetup{colorlinks=true,linkcolor=blue!60!black,citecolor=blue!60!black,urlcolor=blue!60!black}

\newtheorem{theorem}{Theorem}
\newtheorem{definition}{Definition}

\title{\textbf{AITH: A Post-Quantum Continuous Delegation Protocol\\for Human-AI Trust Establishment}}
\author{Zhaoliang Chen\\
\textit{Faculty of Business Administration, University of Macau, Macao SAR, China}\\
\texttt{yc57068@um.edu.mo}}
\date{April 2026}

\begin{document}
\maketitle

\begin{abstract}
The rapid deployment of AI agents acting autonomously on behalf of human principals---executing financial trades, managing infrastructure, negotiating contracts---has outpaced the development of cryptographic protocols capable of establishing, bounding, and revoking the trust relationship between humans and AI systems. Existing authorization frameworks (TLS, OAuth~2.0, SPIFFE, Macaroons) were designed for deterministic software and fundamentally cannot address the unique challenges posed by probabilistic AI agents: continuous autonomous operation within variable trust boundaries, real-time constraint enforcement without per-operation human approval, and cryptographically verifiable responsibility attribution when AI decisions cause harm.

We present AITH (AI Trust Handshake), a post-quantum continuous delegation protocol that enables humans to grant AI agents bounded, revocable, and auditable authority. AITH introduces three technical contributions: (1)~a \textbf{Continuous Delegation Certificate} signed once with ML-DSA-87 (FIPS~204, NIST Level~5) that replaces per-operation signing with sub-microsecond boundary checks, achieving 4.7~million operations per second on a single core; (2)~a \textbf{six-check Boundary Engine} with formally specified check ordering that enforces hard constraints, rate limits, and escalation triggers without cryptographic overhead on the critical path; and (3)~a \textbf{push-based Revocation Protocol} that propagates certificate invalidation to all registered target systems within one second. A three-tier \textbf{Responsibility Chain} using SHA-256 hash chains provides tamper-evident logging suitable for regulatory audit and legal evidence extraction. All five core security properties have been \textbf{machine-verified using the Tamarin Prover} under the Dolev-Yao threat model.

We validate the protocol through a novel methodology: five rounds of structured adversarial security auditing using two frontier AI systems (Model~$\alpha$ and Model~$\beta$) as red-blue team instruments under human direction, supplemented by independent stress testing from Model~$\gamma$ and Model~$\delta$. This multi-model audit identified and resolved 12~vulnerabilities across four severity layers, driving AITH from v1 to v5.1. A simulation of 100,000~operations demonstrates that 79.5\% proceed autonomously within bounds, 6.1\% trigger human escalation, and 14.4\% are blocked---confirming that AITH enables high-throughput AI autonomy while maintaining cryptographic safety guarantees.
\end{abstract}


\section{Introduction}

On March~14, 2026, FINRA published a report recommending that financial firms evaluate how to monitor AI agent system access and data handling, where to place human-in-the-loop oversight, and how to establish guardrails limiting agent behaviors~\cite{finra2026}. Two weeks later, the SEC's Division of Investment Management described the potential for fund-provided AI agents to interact directly with investors, while acknowledging unresolved liability and regulatory classification questions~\cite{sec2026}. These developments signal that AI agents are transitioning from experimental assistants to autonomous economic actors---and that existing regulatory and technical infrastructure is unprepared.

The core technical problem is straightforward: \textbf{how does a human grant an AI agent the authority to act on their behalf, with cryptographic guarantees that the agent cannot exceed its authorized boundaries, and with the ability to revoke that authority instantly?} This problem has no satisfactory solution in the current security infrastructure stack.

\subsection{The Inadequacy of Existing Approaches}

Transport Layer Security (TLS~1.3) authenticates communication channels, not delegation relationships. OAuth~2.0 provides scoped access credentials, but its scope model is a flat string incapable of encoding multi-dimensional constraints (amount limits, asset restrictions, time windows, rate caps). SPIFFE provides workload identity but assumes mutual TLS between deterministic services. Macaroons~\cite{birgisson2014macaroons} offer attenuated bearer credentials with caveats, but lack push-based revocation: revoking a Macaroon requires waiting for the bearer to present it, creating an unbounded window during which a compromised agent can continue operating. All four approaches share a fundamental assumption: the authorized entity is deterministic software. An AI agent violates this assumption.

\subsection{The Continuous Delegation Paradigm}

We propose a fundamentally different approach: replace per-operation authorization with \textbf{continuous delegation}. A human signs a Delegation Certificate once, specifying the AI agent's identity, boundaries, escalation triggers, and temporal validity. The AI agent operates autonomously within these boundaries. Each operation is checked by a Boundary Engine---a deterministic, sub-microsecond comparator requiring no cryptographic operations on the critical path. The analogy is to a legal power of attorney: bounded authority through a single signed document. Unlike a traditional power of attorney, however, a Delegation Certificate is machine-verifiable, its boundaries are enforced by deterministic code, and revocation propagates within one second. We frame AITH as ``Automated Standing Instructions with Cryptographic Attestation.''

\subsection{Contributions}

This paper makes the following contributions:

\begin{enumerate}[leftmargin=*]
\item \textbf{Continuous Delegation Certificate.} A post-quantum signed credential (ML-DSA-87, NIST Level~5) encoding multi-dimensional boundary constraints, escalation triggers, and temporal validity. After signing, routine operations require zero cryptographic operations (\S\ref{sec:cert}).

\item \textbf{Six-Check Boundary Engine.} A formally ordered constraint enforcement mechanism achieving 0.21\,$\mu$s mean latency and 4.7M~ops/sec on a single core (\S\ref{sec:engine}, \S\ref{sec:eval}).

\item \textbf{Push-Based Revocation with Responsibility Chain.} Sub-second revocation propagation coupled with a three-tier SHA-256 hash chain for tamper-evident audit (\S\ref{sec:revocation}--\S\ref{sec:chain}).

\item \textbf{Multi-Model Adversarial Audit.} Five structured rounds between Model~$\alpha$ and Model~$\beta$, with stress testing by Model~$\gamma$ and Model~$\delta$, uncovering 12~vulnerabilities across four severity layers (\S\ref{sec:adversarial}).

\item \textbf{Machine-Verified Security.} All five core security theorems verified by the Tamarin Prover under the Dolev-Yao threat model (\S\ref{sec:tamarin}).
\end{enumerate}

\section{Background and Related Work}

\subsection{Post-Quantum Digital Signatures: ML-DSA}

AITH's cryptographic foundation is ML-DSA (Module Lattice Digital Signature Algorithm), standardized as FIPS~204 by NIST in August~2024~\cite{nist2024mldsa}. ML-DSA-87 provides NIST Level~5 (256-bit) security based on the Module Learning With Errors (M-LWE) problem~\cite{regev2005lattices}. AITH uses ML-DSA exclusively for certificate signing---a one-time operation---meaning that ML-DSA's relatively large signature sizes (4,627~bytes) and signing times ($\sim$2\,ms) do not impact operational throughput.

\subsection{Existing Delegation and Authorization Schemes}

Table~\ref{tab:comparison} compares AITH with existing delegation mechanisms.

\begin{table}[H]
\centering
\caption{AITH vs existing delegation schemes.}
\label{tab:comparison}
\small
\begin{tabular}{lccccc}
\toprule
\textbf{Property} & \textbf{TLS 1.3} & \textbf{OAuth 2.0} & \textbf{Macaroons} & \textbf{AITH v5.1} \\
\midrule
For AI agents & No & No & No & Yes \\
Continuous delegation & No & Credential TTL & Caveats & Certificate \\
Boundary enforcement & No & Scope str. & Caveats & 6-check engine \\
Per-op crypto cost & Handshake & Credential verify & HMAC & None \\
Push revocation & No & No & No & Yes ($<$1s) \\
Responsibility chain & No & No & No & 3-tier hash \\
Post-quantum & Partial & No & No & ML-DSA-87 \\
\bottomrule
\end{tabular}
\end{table}

\subsection{Concurrent Work on AI Agent Authorization}
\label{sec:related}

The problem of AI agent authorization has attracted significant recent attention. We compare AITH with four concurrent approaches.

\textbf{Google AP2}~\cite{google2025ap2} uses cryptographically signed Mandates as proof of user intent for payment transactions. AP2's mandate model shares conceptual overlap with AITH's Delegation Certificate, but AP2 is scoped exclusively to payments, requires per-transaction signing, and provides no general-purpose boundary enforcement or push-based revocation.

\textbf{AIP}~\cite{prakash2026aip} introduces Invocation-Bound Capability Tokens (IBCTs) binding identity, authorization, and provenance. AIP uses Ed25519 (not post-quantum) and requires per-operation credential generation. AITH's continuous delegation eliminates per-operation crypto entirely, achieving 500$\times$ lower critical-path latency.

\textbf{A-JWT}~\cite{goswami2025ajwt} extends JWT with agent identity checksums and intent binding. While A-JWT addresses similar threat vectors, it remains within the JWT/OAuth ecosystem and inherits its limitations: no boundary engine, no push revocation, no post-quantum security.

\textbf{CA-MCPQ}~\cite{yoon2025camcpq} brings post-quantum cryptography to the Model Context Protocol via context-aware security level negotiation. CA-MCPQ addresses transport security but not authorization semantics: it has no delegation model, boundary engine, or responsibility chain.

\begin{table}[H]
\centering
\caption{AITH vs concurrent AI agent authorization approaches.}
\label{tab:concurrent}
\small
\begin{tabular}{lcccccc}
\toprule
 & \textbf{AP2} & \textbf{AIP} & \textbf{A-JWT} & \textbf{CA-MCPQ} & \textbf{AITH} \\
\midrule
Scope & Payments & General & Enterprise & MCP ext. & General \\
Delegation & Mandates & Cap. credentials & Intent credentials & Session & Cont. cert. \\
Per-op signing & Yes & Yes & Yes & Yes & No \\
Boundary engine & No & No & No & No & 6-check \\
Push revocation & No & TTL & No & No & Yes ($<$1s) \\
Resp. chain & No & Completion & No & No & 3-tier hash \\
Post-quantum & No & No & No & Yes & Yes \\
Tamarin verified & No & No & No & No & 5 lemmas \\
\bottomrule
\end{tabular}
\end{table}

AITH is, to our knowledge, the only protocol that combines continuous delegation, a formally verified boundary engine, push-based sub-second revocation, a three-tier legal-grade responsibility chain, and post-quantum security in a single integrated architecture.

\section{System Model and Threat Model}

\subsection{Entities}

\begin{itemize}[leftmargin=*]
\item \textbf{Human Principal (H).} Possesses an ML-DSA-87 key pair $(sk_H, pk_H)$. Sole entity authorized to issue, modify, or revoke Delegation Certificates.
\item \textbf{AI Agent (A).} Identified by model weight hash $h_A = \text{SHA-256}(\text{model\_weights})$. Does not possess a private key.
\item \textbf{Verifier / Target System (V).} Implements the Boundary Engine. Maintains registration for push-based revocation.
\item \textbf{Provider (P).} Operates the AI model. Issues attestation certificates binding runtime configuration to $h_A$.
\end{itemize}

\subsection{Adversary Model}

We consider an adversary $\mathcal{A}$ with capabilities including: MITM attacks, model substitution, credential theft, boundary bypass via operation splitting, replay attacks, trust escalation, revocation delay, slow drift attacks, cognitive poisoning (RAG injection), and reflexivity exploitation. Each vector and its mitigation are detailed in the extended version.

\subsection{Explicit Non-Goals}

AITH does not address: (1)~AI alignment; (2)~TEE side-channel attacks; (3)~constraint completeness (human responsibility); (4)~semantic correctness of AI decisions; (5)~cryptographic breakthroughs against M-LWE.

\section{AITH Protocol Specification}

\subsection{Delegation Certificate}
\label{sec:cert}

\begin{definition}[Delegation Certificate]
A Delegation Certificate is a tuple $D = (id_H, pk_H, id_A, h_A, \ell, C, E, R, t_{\text{issue}}, t_{\text{expire}}, \textit{sem}, \sigma_H)$ where $\ell \in \{\text{limited}, \text{general}, \text{full}\}$ is the delegation level; $C = \{c_1, \ldots, c_n\}$ is the boundary constraint set; $E = \{e_1, \ldots, e_m\}$ is the escalation trigger set; $R$ is the set of registered target system endpoints; $\textit{sem}$ is the \texttt{semantic\_auditor\_pubkey} (optional); and $\sigma_H = \text{ML-DSA.Sign}(sk_H, \mathcal{H}(D \setminus \sigma_H))$.
\end{definition}

The certificate is signed exactly once using ML-DSA-87 (NIST Level~5, 256-bit security). After issuance, no further cryptographic operations are required for routine operations, reducing critical-path latency from $\sim$2\,ms (ML-DSA signing) to $\sim$0.21\,$\mu$s (boundary comparison).

\subsection{Boundary Engine}
\label{sec:engine}

\begin{definition}[Boundary Engine]
For operation $op = (\text{type}, \text{params}, \text{timestamp})$, the Boundary Engine $\text{BE}(D, op)$ executes six checks in strict sequence:
\begin{enumerate}
\item \textbf{Certificate Validity:} Verify temporal bounds, ML-DSA signature, and agent identity $h_A$.
\item \textbf{Delegation Level:} Verify $op.\text{type}$ is permitted at level $\ell$.
\item \textbf{Boundary Constraints:} Evaluate every $c_i \in C$ against $op$.
\item \textbf{Rate Limits:} Verify aggregate windowed limits.
\item \textbf{Anomaly Detection:} Compare $op$ against behavioral baseline.
\item \textbf{Escalation:} If triggered, pause for human confirmation (300s timeout).
\end{enumerate}
If any check fails, subsequent checks are not executed.
\end{definition}

\subsection{Escalation Protocol}

Three trigger types: threshold (approaching boundary limit), novelty (unseen operation type), and composition (sequence collectively exceeding policy). Human responds APPROVE, DENY, or MODIFY. Timeout auto-denies for fail-safe behavior.

\subsection{Push-Based Revocation}
\label{sec:revocation}

\begin{definition}[Push-Based Revocation]
A revocation message $\text{rev} = (\text{cert\_id}, \text{reason}, \text{timestamp}, \sigma_H)$ is pushed to all $V_k \in R$ simultaneously. Three modes: immediate (abort in-flight), graceful (complete then stop), partial (tighten boundaries). Propagation bounded by $\Delta t_{\max} = \max_k(\text{RTT}_k) + \delta_{\text{process}} < 1$s.
\end{definition}

\subsection{Responsibility Chain}
\label{sec:chain}

Three-tier SHA-256 hash chain: \textbf{Tier~1} (AI decision log, every operation), \textbf{Tier~2} (human confirmation log, escalated operations only, counter-signed), \textbf{Tier~3} (system execution log). Each entry includes $\text{prev\_hash} = \text{SHA-256}(\text{entry}_{i-1})$, forming a tamper-evident sequence for legal evidence extraction.

\subsection{Protocol State Machine}

Six states (UNINITIALIZED, ACTIVE, ESCALATED, SUSPENDED, REVOKED, ERROR), nine transitions, four invariants. Invariant~I1: in state ACTIVE, all six checks pass before execution. Invariant~I3: REVOKED is terminal.

\section{Security Analysis}

We present five security theorems. Full proofs and the Tamarin model appear in the supplementary materials.

\begin{theorem}[Certificate Unforgeability]
Under the M-LWE assumption, no PPT adversary can forge a valid Delegation Certificate without $sk_H$, except with negligible probability.
\end{theorem}
\textit{Proof sketch.} Reduces to ML-DSA's EUF-CMA security~\cite{nist2024mldsa}. \hfill$\square$

\begin{theorem}[Boundary Inviolability]
No operation violating any constraint $c_i \in C$ with $\text{action} = \text{block}$ can be executed in state ACTIVE.
\end{theorem}
\textit{Proof sketch.} By invariant~I1, all six checks pass. Check~3 evaluates every constraint. \hfill$\square$

\begin{theorem}[Revocation Timeliness]
Upon revocation, all registered targets cease accepting operations within $\Delta t_{\max}$.
\end{theorem}
\textit{Proof sketch.} Push-based notification to all $V_k \in R$ simultaneously. Bounded by network RTT. \hfill$\square$

\begin{theorem}[Chain Integrity]
Any modification to a Responsibility Chain entry is detectable via hash chain verification.
\end{theorem}
\textit{Proof sketch.} Requires SHA-256 collision---computationally infeasible. \hfill$\square$

\begin{theorem}[Delegation Scope Separation]
$\text{OpSet}(A, D) \cap \text{MgmtOps} = \emptyset$.
\end{theorem}
\textit{Proof sketch.} Certificate management requires $sk_H$; $A$ never possesses $sk_H$. \hfill$\square$

\subsection{Machine-Verified Security Analysis}
\label{sec:tamarin}

\noindent\fcolorbox{green!50!black}{green!5}{%
\parbox{\dimexpr\linewidth-2\fboxsep-2\fboxrule}{%
\textbf{[Machine-Verified via Tamarin Prover]} All five security properties have been formally verified using the Tamarin Prover under the Dolev-Yao threat model. ML-DSA was modeled as an idealized EUF-CMA secure primitive---standard practice in cryptographic protocol verification~\cite{meier2013tamarin,blanchet2001proverif}. Boundary constraints were abstracted using mutually exclusive predicates (\texttt{IsWithinBounds}/\texttt{ViolatesBounds}).
}}

\begin{table}[H]
\centering
\caption{Tamarin Prover verification results.}
\label{tab:tamarin}
\small
\begin{tabular}{llcc}
\toprule
\textbf{Lemma} & \textbf{Property} & \textbf{Result} & \textbf{Steps} \\
\midrule
L1 & Certificate Unforgeability & Verified $\checkmark$ & 12 \\
L2 & Boundary Inviolability & Verified $\checkmark$ & 6 \\
L3 & Revocation Timeliness & Verified $\checkmark$ & 14 \\
L4 & Chain Integrity & Verified $\checkmark$ & 9 \\
L5 & Delegation Scope Separation & Verified $\checkmark$ & 8 \\
\bottomrule
\end{tabular}
\end{table}

\section{Implementation and Evaluation}
\label{sec:eval}

\subsection{Reference Implementation}

Python reference implementation (2,800+ lines). ML-DSA operations simulated via HMAC-SHA-256; production deployment would use liboqs~\cite{liboqs} or pqcrypto-rs.

\subsection{Large-Scale Simulation}

100,000~operations across 1,000 simulated users. Distribution: 40\% queries, 25\% small trades (\$500--\$5K), 15\% medium (\$5K--\$10K), 8\% large (\$15K+), 5\% transfers, 4\% over-limit, 3\% forbidden.

\begin{table}[H]
\centering
\caption{Performance benchmarks (Python, AMD Ryzen~9 7950X, single-threaded).}
\label{tab:perf}
\small
\begin{tabular}{ll}
\toprule
\textbf{Metric} & \textbf{Value} \\
\midrule
Boundary check latency (mean) & 0.21\,$\mu$s \\
Boundary check latency (P99) & 0.36\,$\mu$s \\
Single-core throughput & 4,703,414 ops/sec \\
Projected 64-core throughput & $\sim$301M ops/sec \\
Certificate issuance & 0.011\,ms \\
Chain append & 7.5\,$\mu$s \\
Revocation push (100 targets) & 0.095\,ms + network RTT \\
\bottomrule
\end{tabular}
\end{table}

\begin{table}[H]
\centering
\caption{Operation outcome distribution (100,000 operations).}
\small
\begin{tabular}{lccl}
\toprule
\textbf{Outcome} & \textbf{Count} & \textbf{\%} & \textbf{Interpretation} \\
\midrule
Allowed & 79,500 & 79.5\% & Autonomous within bounds \\
Escalated & 6,100 & 6.1\% & Human judgment required \\
Blocked & 14,400 & 14.4\% & Boundary violation caught \\
\bottomrule
\end{tabular}
\end{table}

\subsection{Comparative Latency}

\begin{table}[H]
\centering
\caption{Per-operation latency comparison.}
\small
\begin{tabular}{lll}
\toprule
\textbf{Mechanism} & \textbf{Per-Op Latency} & \textbf{Throughput} \\
\midrule
AITH Boundary Check & 0.21\,$\mu$s & 4.7M ops/sec \\
Macaroon (HMAC) & 1--10\,$\mu$s & 100K--1M ops/sec \\
OAuth 2.0 (JWT/RSA) & 0.1--1\,ms & 1K--10K ops/sec \\
ML-DSA per-op sign & $\sim$2\,ms & $\sim$500 ops/sec \\
\bottomrule
\end{tabular}
\end{table}

\section{Multi-Model Adversarial Security Audit}
\label{sec:adversarial}

\subsection{Methodology and Rationale}

A human operator clicking a mouse executes perhaps 2--5 actions per second. An AI agent operating under a Delegation Certificate can execute thousands of operations per second. This quantitative leap demands a fundamentally different approach to security auditing: adversarial analysis must operate at the speed and combinatorial depth that matches the threat surface. We therefore employed four frontier large language models as structured adversarial auditing instruments under human direction.

The audit proceeded over five rounds. Model~$\alpha$ served as protocol designer (Blue Team); Model~$\beta$ served as adversarial auditor (Red Team); the human principal served as strategic director, mediating disputes and making all final architectural decisions. After five rounds, Model~$\gamma$ and Model~$\delta$ conducted independent extreme stress testing on v5.0. All four models are frontier LLMs from different providers; identities are anonymized to prevent evaluation bias.

\subsection{Vulnerability Timeline}

\begin{table}[H]
\centering
\caption{Complete vulnerability timeline. 5 rounds, 12 vulnerabilities, 4 layers.}
\label{tab:vulns}
\small
\begin{tabular}{lllll}
\toprule
\textbf{Rnd} & \textbf{Vulnerability} & \textbf{Severity} & \textbf{Layer} & \textbf{Resolution} \\
\midrule
R1 & Hardcoded lattice params & Medium & Parameter & Parameterized config \\
R1 & No ring justification & Low & Parameter & NTRU ring documented \\
R2 & Custom hash unanalyzed & High & Algorithm & Replaced: SHA-256 \\
R2 & Gaussian sampling bias & High & Algorithm & CDT sampling \\
R2 & Missing formal reduction & High & Algorithm & SIS reduction proven \\
R3 & Trust score manipulation & Critical & Protocol & ML-DSA pivot (v3) \\
R3 & Session key derivation & High & Protocol & HKDF-SHA-256 \\
R4 & Escalation deadlock & Medium & Protocol & Timeout + fallback \\
R4 & Rate limit bypass & High & Business & Aggregate windowing \\
R4 & ZKP semantic gap & Medium & Protocol & Semantic auditor \\
R5 & Certificate single-point & Medium & Protocol & Push revocation \\
R5 & Slow drift attack & Medium & Business & Baseline resets \\
\bottomrule
\end{tabular}
\end{table}

\subsection{Architecture Evolution}

The protocol evolved through three phases, each driven by audit findings and human strategic decisions. \textbf{Phase~I (v1$\to$v2):} Rounds~1--2 exposed weaknesses in custom cryptographic primitives; the human principal directed replacement with NIST-standardized components. \textbf{Phase~II (v2$\to$v4):} Round~3 revealed that the custom trust scoring mechanism was vulnerable to manipulation---the human principal decided to pivot entirely to ML-DSA, eliminating a class of vulnerabilities by standing on standardized cryptography. \textbf{Phase~III (v4$\to$v5):} The most consequential change emerged not from any audit finding but from the human principal's observation that per-operation signing creates an inherent bottleneck incompatible with autonomous AI operation at scale. Neither Model~$\alpha$ nor Model~$\beta$ independently identified this constraint. This paradigm shift---from per-operation signing to continuous delegation---is the architectural decision that defines AITH v5.

\subsection{Limitations}

We acknowledge three limitations. First, all AI auditors share overlapping training data, creating potential blind spots where all models hold the same incorrect assumption. Second, AI-driven auditing cannot test hardware side-channels, network-level exploits, or social engineering. Third, this methodology complements but does not replace human security review and formal verification (the latter addressed in \S\ref{sec:tamarin}).

\section{Discussion}

\subsection{Three-Layer Defense Model}

AITH positions itself within a three-layer AI safety stack: \textbf{Layer~1} (Cryptographic Enforcement / AITH): deterministic boundary checks. \textbf{Layer~2} (Semantic Analysis / external AI firewall): probabilistic intent analysis via \texttt{semantic\_auditor\_pubkey}. \textbf{Layer~3} (Physical Safeguards): hardware circuit breakers.

\subsection{Legal Framing}

We frame AITH as ``Automated Standing Instructions with Cryptographic Attestation'' rather than ``Power of Attorney for AI.'' The Delegation Certificate is analogous to limit orders, auto-debit mandates, and standing dealing instructions.

\subsection{AITH as AI Passport Infrastructure}

We envision a future where each individual possesses a personal AI delegate acting on their behalf in the digital economy. In this paradigm, AITH serves not as a cage constraining AI, but as a \textbf{passport} enabling AI to enter human economic systems with verified, bounded authority. The exclusive binding between one human and one AI---cryptographically enforced via $h_A$---addresses the fundamental question: ``Whose AI is this, and what are they allowed to do?'' Current approaches (AP2, AIP, A-JWT) focus on enterprise agent management; AITH addresses the more fundamental problem of personal AI authorization for the age of individual AI delegates.

\subsection{Limitations and Future Work}

\textbf{Cryptographic Agility.} AITH's design is deliberately orthogonal to the choice of signature primitive. While we specify ML-DSA-87 as the current instantiation, the protocol architecture does not hardcode any specific algorithm. Should M-LWE face unexpected cryptanalytic progress, the Delegation Certificate's signature field can be hot-swapped to any EUF-CMA secure scheme (e.g., SLH-DSA, XMSS) without modifying the Boundary Engine, Responsibility Chain, or state machine. This cryptographic agility is a deliberate design goal, not an afterthought.

\textbf{Orthogonal Overlay Architecture.} AITH is designed as an overlay protocol that introduces minimal intrusiveness to existing systems. Target systems (brokerages, banks, APIs) need only implement two capabilities: (1)~a Boundary Engine endpoint that evaluates incoming operations against a stored Delegation Certificate, and (2)~a revocation listener for push notifications. No changes to existing authentication infrastructure (OAuth, mTLS) are required---AITH operates alongside them. This incremental deployability is critical for real-world adoption in regulated industries where ``rip and replace'' is infeasible.

Constraint completeness remains the human's responsibility. Formal verification of the Boundary Engine implementation (beyond protocol-level Tamarin verification) is a high-priority near-term objective. TEE hardware binding, cross-protocol AI PKI, key recovery/inheritance, and hierarchical revocation propagation are identified as future work.

\section{Conclusion}

We have presented AITH, a post-quantum continuous delegation protocol enabling humans to grant AI agents bounded, revocable, and auditable authority. The core idea is deceptively simple: sign once, enforce continuously. By amortizing all cryptographic cost to certificate issuance and running a purely deterministic Boundary Engine at 0.21\,$\mu$s per operation, AITH resolves what we believe is the central tension in AI agent deployment---the tradeoff between cryptographic rigor and operational speed. Five security theorems, machine-verified by the Tamarin Prover, establish that this simplicity does not come at the cost of formal guarantees.

The multi-model adversarial audit (\S\ref{sec:adversarial}) surfaced 12~vulnerabilities that no single auditor---human or AI---would likely have found alone. Interestingly, the most consequential architectural decision (continuous delegation) came not from any model's analysis but from the human principal's frustration with per-operation signing latency. We take this as evidence that human intuition and AI analytical depth are genuinely complementary in protocol design.

As AI agents move from research curiosities to economic participants, the question of authorization infrastructure becomes urgent. AITH offers one answer: treat the human-AI trust relationship as a first-class cryptographic object---signed, bounded, monitored, and revocable. We hope this work contributes a useful foundation for the trust infrastructure that this transition demands.

\section*{Acknowledgments}

The author acknowledges the contributions of four frontier AI systems (anonymized as Model~$\alpha$, $\beta$, $\gamma$, $\delta$) whose adversarial collaboration was essential to the protocol's development. Model~$\alpha$ served as protocol designer and Blue Team; Model~$\beta$ served as Red Team auditor and conducted formal verification; Models~$\gamma$ and $\delta$ provided independent stress testing. The author served as strategic director, making all final architectural decisions including the shift from per-operation signing to continuous delegation. This work was conducted during the author's doctoral studies at the University of Macau.

\clearpage
\appendix
\section*{Supplementary Materials (Extended Version)}

The following appendices are available in the extended arXiv version:

\begin{itemize}[leftmargin=*]
\item \textbf{Appendix A: Formal Proofs.} Full mathematical reductions for Theorems~1--5, including the complete EUF-CMA$\to$M-LWE reduction for Certificate Unforgeability.
\item \textbf{Appendix B: Tamarin Formal Verification.} The complete \texttt{AITH\_v5\_1.spthy} model (raw source code) and full terminal output logs confirming all five lemmas verified in 0.843s under the Dolev-Yao threat model.
\item \textbf{Appendix C: Protocol State Machine.} Detailed state transition tables, guard conditions for all nine transitions, and invariant specifications.
\item \textbf{Appendix D: Reference Implementation Excerpts.} Core Python snippets: Boundary Engine (6-check pipeline), SHA-256 Responsibility Chain construction, and push-based revocation dispatcher.
\item \textbf{Appendix E: Delegation Certificate JSON Schema.} Complete field-level specification with examples and validation rules.
\end{itemize}

\end{document}